\documentclass[a4paper,UKenglish,cleveref, autoref, thm-restate]{lipics-v2021}

\usepackage{amsmath,amssymb,amsthm}
\usepackage{graphicx}
\usepackage{algorithm}
\usepackage{algpseudocode}
\usepackage{dsfont}
\usepackage{mathtools}
\usepackage{tikz-cd}
\usetikzlibrary{decorations.pathmorphing}
\usepackage{enumitem}   
\usepackage[size=tiny]{todonotes}
\usepackage{todonotes}

\newcommand{\Prob}[1]{\mathrm{Pr}\left[#1\right]}

\newcommand{\eps}{\varepsilon}

\DeclarePairedDelimiter\floor{\lfloor}{\rfloor}
\DeclarePairedDelimiter\ceil{\lceil}{\rceil}
\DeclareMathOperator{\dist}{dist}

\title{On the Giant Component of Geometric Inhomogeneous Random Graphs} 

\titlerunning{Giant Component of Geometric Inhomogeneous Random Graphs} 

\author{Thomas Bl{\"a}sius}{Karlsruhe Institute of Technology}{thomas.blaesius@kit.edu}{}{}

\author{Tobias Friedrich}{Hasso Plattner Institute, University of Potsdam, Germany}{Tobias.Friedrich@hpi.de}{{https://orcid.org/0000-0003-0076-6308}}{}

\author{Maximilian Katzmann}{Karlsruhe Institute of Technology}{maximilian.katzmann@kit.edu}{}{}

\author{Janosch Ruff}{Hasso Plattner Institute, University of Potsdam, Germany}{Janosch.Ruff@hpi.de}{}{}

\author{Ziena Zeif}{Hasso Plattner Institute, University of Potsdam, Germany}{Ziena.Zeif@hpi.de}{https://orcid.org/0000-0003-0378-1458}{}

\authorrunning{T.~Bl{\"a}sius, T.~Friedrich, M.~Katzmann, J.~Ruff, Z.~Zeif} 

\Copyright{Thomas Bl{\"a}sius, Tobias Friedrich, Maximilian Katzmann, Janosch Ruff, Ziena Zeif} 

\ccsdesc[500]{Theory of computation~Random network models}

\keywords{geometric inhomogeneous random graphs, connectivity, giant
  component} 

\category{} 

\relatedversion{} 

\nolinenumbers 

\hideLIPIcs  

\EventEditors{John Q. Open and Joan R. Access}
\EventNoEds{2}
\EventLongTitle{31st Annual European Symposium on Algorithms}
\EventShortTitle{ESA 2023}
\EventAcronym{ESA}
\EventYear{2023}
\EventDate{September 4--6, 2023}
\EventLocation{Amsterdam, The Netherlands}
\EventLogo{}
\SeriesVolume{42}
\ArticleNo{23}

\begin{document}

\maketitle
	
\begin{abstract}
  In this paper we study the threshold model of \emph{geometric inhomogeneous random graphs}
  (GIRGs); a generative random graph model that is closely related to
  \emph{hyperbolic random graphs} (HRGs).  These models have been
  observed to capture complex real-world networks well with respect to
  the structural and algorithmic properties.  Following comprehensive
  studies regarding their \emph{connectivity}, i.e., which parts of
  the graphs are connected, we have a good understanding under which
  circumstances a \emph{giant} component (containing a constant
  fraction of the graph) emerges.

  While previous results are rather technical and challenging to work
  with, the goal of this paper is to provide more accessible proofs.
  At the same time we significantly improve the previously known
  probabilistic guarantees, showing that GIRGs contain a giant
  component with probability $1 - \exp(-\Omega(n^{(3-\tau)/2}))$ for
  graph size $n$ and a degree distribution with power-law exponent
  $\tau \in (2, 3)$.  Based on that we additionally derive insights
  about the connectivity of certain induced subgraphs of GIRGs.
\end{abstract}
	
\newpage
	
\section{Introduction}

Geometric inhomogeneous random graphs (GIRGs) are a generative graph
model where vertices are weighted and placed in a geometric ground
space and the probability for two of them to be adjacent depends on
the product of their weights, as well as their
distance~\cite{k-girggcg-18}.  In a sense the model combines the
strengths of \emph{inhomogeneous random graphs}~\cite{s-g-02} and
\emph{random geometric graphs}~\cite{p-rgg-03}.  Introduced as a
simplified and more general version of \emph{hyperbolic random graphs}
(HRGs)~\cite{kpk-h-10}, GIRGs share crucial properties with complex
real-world networks.  Such networks are typically characterized by a
\emph{heterogeneous degree distribution} (with few high-degree
vertices, while the majority of vertices has small degree), \emph{high
  clustering} (vertices with common neighbors are likely adjacent
themselves), and a \emph{small diameter} (longest shortest path), and
it has been shown that GIRGs and HRGs capture these properties
well~\cite{gpp-rhg-12, k-girggcg-18, ms-k-19}.

Beyond these structural properties, GIRGs have also been observed to
be a good model for real-wold networks when it comes to the
performance of graph algorithms~\cite{bf-evacaga-22}.  This makes the
GIRG framework relevant for algorithmic purposes in multiple ways.  On
the one hand, they are a useful tool in the context of average-case
analysis, where they yield more realistic instances than, e.g., the
Erdős--Rényi model, while it is still sufficiently simple to be
mathematically accessible~\cite{bf-evacaga-22, bkl-g-22}.  On the
other hand, we can use GIRGs to generate an abundance of benchmark
instances with varying properties, allowing us to perform thorough
evaluations of algorithms even when real-world data is
scarce~\cite{bfw-ecmfsfn-21, bfk-eggihrg-22}.

One of the most basic graph properties, which is also relevant from an
algorithmic point of view, is \emph{connectivity}, i.e., the question
about what parts of a graph are connected via paths.  For random
graphs, the first question that typically arises in the context of
connectivity revolves around the emergence of a so-called \emph{giant
  component}, which is a connected component whose size is linear in
the size of the graph.  The existence of a giant has been researched
on many related graph models like \emph{Erdős--Rényi random
  graphs}~\cite{er-erg-60, er-rgi-59}, \emph{random geometric
  graphs}~\cite{ar-cgup-02, dmp-cdrgg-08, p-rgg-03, gk-c-98}, as well
as on \emph{Chung-Lu random graphs} that also capture inhomogeneous
random graphs~\cite{acl-rgmplg-01, cl-adrgged-02, cl-ccrgg-02}.

Unsurprisingly, being such a fundamental feature, connectivity has
also been studied on GIRGs, and since HRGs are so closely related to
them, we consider the corresponding results to be relevant here as
well.  For HRGs we know how the emergence of a giant depends on
certain model properties that control the degrees of the resulting
graph~\cite{bfm-giant-15, fm-giant-18}.  We note that some analyses
there are based on a coupling from HRGs to a continuum percolation
model that exhibits a strong resemblance to GIRGs
(see~\cite[Section~2]{fm-giant-18} and~\cite[Part~I,
Section~3.5]{k-girggcg-18}).  Beyond the giant we also have bounds on
the size of the second largest component of HRGs~\cite{km-slcrhg-19}.
For GIRGs it is known that a giant exists \emph{with high
  probability}, i.e., with probability $1 - O(1/n)$, where $n$ denotes
the number of vertices in the graph.

In this paper, we show that the threshold GIRGs have a giant component
with probability at least $1 - \exp(-\Omega(n^{(3 - \tau)/2}))$.  This
improves the previous results in two ways.  First, our proof is
simpler and shorter than the technical existing proofs for
HRGs~\cite{bfm-giant-15, fm-giant-18}.  Secondly, our probability
bound is substantially stronger compared to previous bounds.
Moreover, we note that our improved bound does not only hold for the
full graph but also translates to subgraphs located in restricted
regions of the ground space.  The argument for this is inspired by a
technique used for HRGs~\cite[Section 4]{fm-giant-18} (though it is
much simpler in our case).

Besides providing more accessible insights in the connectivity of
GIRGs, we believe that our results, in particular those on subgraphs
in restricted regions, can be helpful for algorithmic applications.
For example in problems like \emph{balanced connected
  partitioning}~\cite{cc-aambcgp-21}, one is interested in
partitioning a graph into connected components of (roughly) equal size
and in \emph{component order connectivity}~\cite{ghi-scocm-13} the
goal is to find a small separator that divides the graph into
components of bounded size.  There it is important, that the graph
cannot only be separated into smaller pieces but that these pieces
remain actually connected.

In the following, we give a brief overview of the basic concepts used
in the paper (Section~\ref{sec:preliminaries}) before presenting our
proofs regarding the emergence of a giant in GIRGs
(Section~\ref{sec:existence}).

\section{Preliminaries}
\label{sec:preliminaries}

\subparagraph{Geometric Inhomogeneous Random Graphs.}

Let $\mathbb B^d = [0, 1]^d$ be the $d$-dimensional hypercube
($\mathbb B$ for ``box'') and let $\dist$ be the $L_\infty$ metric,
i.e., for $x = (x_1, \dots, x_d) \in \mathbb B^d$ and
$y = (y_1, \dots, y_d) \in \mathbb B^d$ we have
$\dist(x, y) = \max_{i \in [d]} |x_i - y_j|$.

A \emph{geometric inhomogeneous random graph (GIRG)} $G = (V, E)$ with
\emph{ground space} $\mathbb B^d$ is obtained in three steps.  The
first step consists of a homogeneous Poisson point process
on~$\mathbb{B}^d$, with an intensity that yields $n$ points in
expectation.  Each point is then considered to be a vertex in the
graph.  In the second step, each vertex $v$ is assigned a
\emph{weight} $w_v > 1$ that is sampled according to a Pareto
distribution with exponent $\tau \in (2, 3)$, i.e.,
$\Prob{w_v \le w} = 1 - w^{-(\tau - 1)}$.  In the third step, any two
vertices $u$ and $v$ are connected by an edge with a probability that
depends on their distance and their weights.  More precisely, there
are two variants.  In a \emph{threshold GIRG}, $u$ and $v$ are
adjacent if and only if
\begin{align*}
  \dist(u, v) \le \left( \frac{\lambda w_u w_v}{n} \right)^{1/d},
\end{align*}
where $\lambda > 0$ controls the expected average degree of the graph.
In the \emph{temperate} variant we have an additional temperature
parameter $T \in (0, 1)$ and the probability for $u$ and $v$ to be adjacent
is given by
\begin{align*}
  \Prob{\{u, v\} \in E} = \min \left\{ 1, \left(\frac{\lambda w_u w_v}{n \cdot (\dist(u, v))^d} \right)^{1/T} \right\}.
\end{align*}
The threshold variant is the limit of the temperature variant for
$T \to 0$.  We denote the resulting probability distribution of graphs
with $\mathcal{G}(n, \mathbb{B}^d, \tau, \lambda, T)$ for general
GIRGs (allowing temperatures in $T \in [0, 1)$).  When we just refer
to the threshold case, we use
$\mathcal{G}(n, \mathbb{B}^d, \tau, \lambda)$.  We assume the
parameters $d, \tau, \lambda$, and $T$ to be constant, i.e.,
independent of $n$.

\subparagraph{GIRG Variants.}

In the literature, several variants of the GIRG model have been
studied and we want to briefly discuss the choice we made here.
Usually, GIRGs are considered with a torus $\mathbb T^d$ as ground
space, i.e., the distance in the $i$th dimension, between $x$ and $y$
is $\min\{|x_i - y_i|, 1 - |x_i - y_i|\}$ instead of just
$|x_i - y_i|$.  The torus usually makes arguments easier as it
eliminates the special case close to the boundary of $\mathbb B^d$.
However, in our case, this is not relevant.  Moreover, as distances in
$\mathbb T^d$ are only smaller than in $\mathbb B^d$, all our results
concerning the largest connected component directly translate to the
case where $\mathbb T^d$ is the ground space.

Moreover, instead of sampling $n$ points uniformly at random in the
ground space, we use a Poisson point process.  This is a technique
often used in geometric random graphs as it makes the number of
vertices appearing in disjoint regions stochastically independent.
This is a similar difference as the one between the Erdős--Rényi model
$G(n, m)$ with a fixed number of edges $m$ and the Gilbert model
$G(n, p)$ with a fixed probability $p$ for each individual edge to
exist.  We generally advocate for using the Poisson variant of the
GIRG model.

\subparagraph{Poisson Point Process.}

Let $R \subseteq \mathbb B^d$ be a region of the ground space with
volume $a$.  Then, the size of the vertex set $V(R)$, i.e., the number
of vertices that are sampled in $R$ is a random variable following a
Poisson distribution with expectation $\mu = an$.  This in particular
means that the probability for $R$ to contain no vertex is
$\exp(-\mu)$.

We note that the Poisson point process we consider is a \emph{marked}
process, where each point sampled from $\mathbb B^d$ obtains a weight
sampled from a weight space $\mathcal W$ as a mark.  Due to the
marking theorem, this is equivalent to considering an (inhomogeneous)
Poisson point process of the product space
$\mathbb B^d \times \mathcal W$, i.e., colloquially speaking, each
point pops up with a position and a weight instead of initially only
having a position and drawing the weight as an afterthought.  This is
also equivalent to just sampling the number of points $N$ following a
Poisson distribution and viewing the positions and the weights as
marks that are sampled subsequently for each of the $N$ points.
Throughout the paper, we switch between these different perspectives
without making this explicit.

\subparagraph{Lowest Weights Dominate.}

We regularly consider weight ranges $[w_1, w_2]$ with
$w_2 \ge c \cdot w_1$ for a constant $c > 1$.  The probability for a
$v$ to have weight in $[w_1, w_2]$ is dominated by $w_1$:
\begin{equation*}
  \Prob{w_v \in [w_1, w_2]}
  = w_1^{- (\tau - 1)} - w_2^{- (\tau - 1)} 
  \ge w_1^{-(\tau - 1)} \cdot \left(1 - c^{-(\tau - 1)} \right)
  \in \Theta(\Prob{w_v \ge w_1}).
\end{equation*}
	
\section{Existence of a Giant Component}
\label{sec:existence}
	
We want to show that a threshold GIRG is highly likely to contain a
connected component of linear size.  Our argument goes roughly as
follows.  We first note that vertices with weight at least
$\sqrt{n / \lambda}$ form a clique, which we call the \emph{core} of
the graph.  For each non-core vertex, we can show that the probability
that it has a path into the core is non-vanishing, i.e., it is lower
bounded by a non-zero constant.  This already shows that we get a
connected graph of linear size in expectation.
	
To show concentration, i.e., that we get a large connected component
with high probability, we essentially need to show that the events for
different low-weight vertices to connect to the core are sufficiently
independent of each other.  To this end, we subdivide the ground space
into a grid of regular \emph{cells} of side length $\Delta$.  We call
a cell \emph{nice} if a linear number of its vertices connect to the
core via paths not leaving the cell and then show that a cell is nice
with non-vanishing probability.  As this only considers paths within
the cell, the different cells are independent.  Thus, we get a series
of independent coin flips, one for each cell.  If a constant fraction
of these coin flips succeeds, we have a connected component of linear
size.  Hence, if the number of cells is sufficiently large, we get
concentration via a Chernoff bound.  It follows that we essentially
want to choose the cell width $\Delta$ to be as small as possible such
that cells are still nice with non-vanishing probability.

In Section~\ref{sec:layer-paths}, we first show that every vertex has
constant probability to have a path to the core.  In fact, we show
something slightly stronger, by considering not just any paths but
so-called layer paths.  Afterwards, we use this result in
Section~\ref{sec:cell-niceness}, to bound the probability for a cell
to be nice.  This then also informs us on how to choose the cell width
$\Delta$ and thus on how many cells we obtain.  With this, we can wrap
up the argument in Section~\ref{sec:giant-lower} by applying a
Chernoff bound.  Besides our main results, we there also mention
immediate implications.
	
\subsection{Layer Paths} \label{sec:layer-paths} We want to show that,
for any individual vertex, the probability that it has a path to a
vertex of the core is non-vanishing.  For this, we define the
$\ell$-th \emph{layer} $V_\ell$ to be the set of vertices with weight
in $[e^{\ell / 2}, e^{(\ell + 1) / 2})$.  Note that the upper and
lower bounds are a constant factor apart and thus (as mentioned in
Section~\ref{sec:preliminaries}) the probability for a vertex to have
layer $\ell$ is asymptotically dominated by the lower bound, i.e.,
$\Prob{v \in V_\ell} \in \Theta(\Prob{w_v \ge e^{\ell / 2}}) =
\Theta(e^{- \ell (\tau - 1) / 2})$.

A path $(v_0, \dots, v_k)$ is a \emph{layer path} if it goes from one
layer to the next in every step, i.e., $v_i \in V_\ell$ implies
$v_{i - 1} \in V_{\ell - 1}$ for every $i \in [k]$.  Note that
vertices in layer $\ceil*{\log(n / \lambda)}$ have weight at least
$\sqrt{n / \lambda}$ and thus belong to the core.  Thus, the following
lemma shows that each vertex has a layer path to the core with
non-vanishing probability.

\begin{lemma}
  \label{lem:layer-path}
  Let $G \sim \mathcal{G}(n, \mathbb B^d, \tau, \lambda)$ be a
  threshold GIRG and let $v$ be a non-core vertex.  The probability
  that there is a layer path from $v$ to layer
  $\ceil*{\log (n / \lambda)}$ is non-vanishing.
\end{lemma}
\begin{proof}
  We bound the probability that such a layer path exists in three
  steps.  First, we bound the probability that a vertex $u$ on layer
  $\ell$ has a neighbor in layer $\ell + 1$.  In the second step, we
  consider the intersection of the events where this happens on all
  considered layers.  Finally, we show that the resulting probability
  is non-vanishing.
	
  For the first step, consider two vertices $u \in V_\ell$ and
  $v \in V_{\ell + 1}$ in consecutive layers, as shown in
  Figure~\ref{fig:layer-path}.  Both their weights are at least
  $w = e^{\ell / 2}$.  Thus, they are definitely adjacent if their
  distance $\dist(u, v)$ satisfies
  \begin{equation*}
    \dist(u, v) \le
    \left( \frac{\lambda w^2}{n} \right)^{1/d} =
    {\lambda^{1 / d}} \left(\frac{e^{\ell}}{n}\right)^{1 / d} 
    \eqqcolon \Delta_\ell.
  \end{equation*}
  If vertex $u \in V_\ell$ is the current vertex from which we want to
  make the next step in a layer path, we are thus interested in the
  probability that there is a vertex $v$ that lies in layer $\ell + 1$
  with $\dist(u, v) \le \Delta_\ell$.  Since these two events (being
  in layer $\ell + 1$ and having sufficiently low distance) are
  independent, the probability that both happen is
  $\Prob{v \in V_{\ell + 1}} \cdot \Prob{\dist(u, v) \le
    \Delta_\ell}$.  As mentioned above, we have
  $\Prob{v \in V_{\ell + 1}} \in \Theta(e^{- \ell (\tau - 1) / 2})$.
  Moreover,
  $\Prob{\dist(u, v) \le \Delta_\ell} \in \Theta(\Delta_\ell^d) =
  \Theta(e^\ell / n)$.  Hence, we obtain
  \begin{align*}
    \Prob{v \in V_{\ell + 1}} \cdot \Prob{\dist(u, v) \le \Delta_\ell} \in
    \Theta\left(e^{- \ell (\tau - 1) / 2} \cdot e^\ell / n\right) =
    \Theta\left(e^{\ell (3 - \tau) / 2} / n\right).
  \end{align*}
  To conclude the first step of the proof, let $X_\ell$ be the number
  of vertices in layer $\ell + 1$ with distance at most $\Delta_\ell$
  to $u \in V_\ell$.  By the above probability, we have
  $\mathbb{E}[{X_\ell}] = \Theta(e^{\ell (3 - \tau) / 2})$.  We
  consider the event $X_\ell > 0$ and call it $A_\ell$.  Note that
  $A_\ell$ implies that $u$ has at least one neighbor in the next
  layer.  As $X_\ell$ follows a Poisson distribution, we get
  \begin{figure}[t]
    \centering \includegraphics{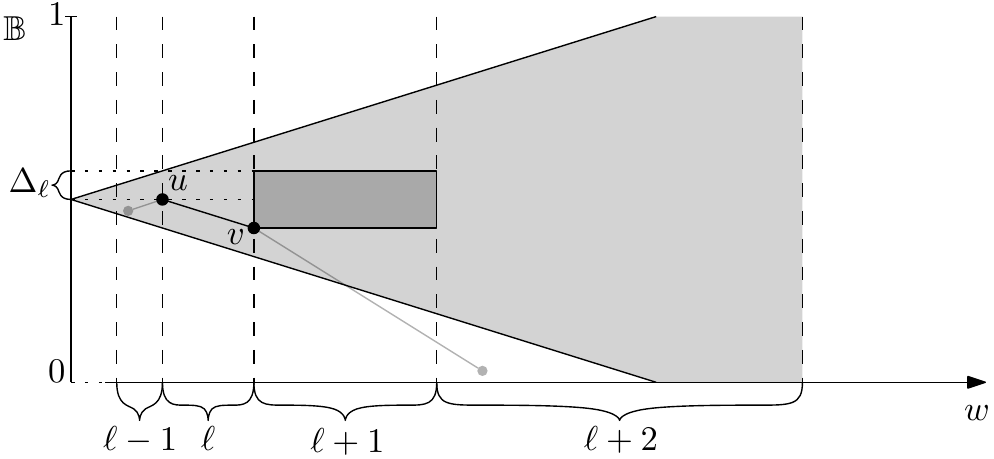}
    \caption{
      Excerpt of a one-dimensional GIRG with the weights on the
      $x$-axis and the ground space $\mathbb{B}$ on the $y$-axis.  A
      layer path spans from layer $\ell - 1$ to $\ell + 2$.  The gray region
      is the neighborhood of vertex $u$.  The dark-gray region
      contains all vertices in layer $\ell + 1$ that have distance at
      most $\Delta_{\ell}$ to $u$.}
    \label{fig:layer-path}
  \end{figure}
  \begin{equation*}
    \Prob{A_\ell}
    = 1 - \Prob{X_\ell = 0}
    = 1 - \exp\left( - \mathbb{E}[X_\ell] \right)
    = 1 - \exp\left( - \Theta\left(e^{\ell (3 - \tau) / 2}\right) \right).
  \end{equation*}
	
  In the second step of the proof, we now consider the intersection of
  all the independent events
  $A_{0}, A_{1}, \dots, A_{\ceil*{\log(n / \lambda)}}$, which is
  sufficient for a layer path starting in layer $0$ to exist.  Note
  that a lower bound for the probability of this intersection also
  gives a lower bound for the existence of a layer path starting in
  any other layer.  To show that this intersection occurs with
  non-vanishing probability, we utilize the fact that $\Prob{A_\ell}$
  approaches~$1$ very quickly for increasing $\ell$.  More precisely,
  we show that for a constant $c$, all subsequent events $A_\ell$ with
  $\ell \ge c$ are sufficiently likely, that we can simply take the
  union bound over their complements.  Thus, we obtain
  \begin{align*}
    \Prob{\bigcap_{\ell = 0}^{\ceil*{\log(n / \lambda)}} A_\ell} =
    \Prob{\bigcap_{\ell = 0}^{c - 1} A_\ell} \cdot \Prob{\bigcap_{\ell = c}^{\ceil*{\log(n / \lambda)}} A_\ell}.
  \end{align*}
  Clearly, the first factor is non-vanishing as it is the product of
  constantly many non-zero constants.  For the second factor, we
  consider the complementary events and apply the union bound to
  obtain
  \begin{align*}
    \Prob{\bigcap_{\ell = c}^{\ceil*{\log(n / \lambda)}} A_\ell}
    &= 1 - \Prob{\bigcup_{\ell = c}^{\ceil*{\log(n / \lambda)}} {A_\ell^C}} \\
    &\ge 1 - \sum_{\ell = c}^{\ceil*{\log(n / \lambda)}} (1 - \Prob{A_\ell}) \\
    &= 1 - \sum_{\ell = c}^{\ceil*{\log(n / \lambda)}} \exp\left(-\Theta(e^{\ell(3 - \tau) / 2})\right).
  \end{align*}
  Since the sum converges, we can choose $c$ to be a sufficiently
  large constant such that the sum is bounded by any constant
  $\eps > 0$.  The above expression is thus at least $1 - \eps$, which
  is non-vanishing.
\end{proof}

Observe that Lemma~\ref{lem:layer-path} already shows that the
expected number of vertices with a layer path to the core is linear.
Thus, the expected size of the connected component including the core
vertices is linear.  To show concentration, we separate the ground
space into cells that are then considered independently.

\subsection{A Coin Flip for Each Cell}
\label{sec:cell-niceness}

We subdivide the ground space into a grid of regular cells of side
length $\Delta$.  We first show that the high-weight vertices of each
cell are likely to induce a connected graph.  This is useful as we can
afterwards focus on vertices of lower weight.  As edges between
low-weight vertices are short, layer paths on these vertices can cover
only a small distance and thus only few of them leave their cell,
which makes different cells (mostly) independent.

\begin{lemma}
  \label{lem:connected-core}
  Let $G \sim \mathcal{G}(n, \mathbb B^d, \tau, \lambda)$ be threshold
  GIRG, and let $C$ be a cell of side length~$\Delta$ and let $w$ be a
  weight.  Then, the graph induced by vertices in $C$ of weight at
  least $w$ is connected with probability at least
  \begin{align*}
    1 - \frac{(2\Delta)^d}{\lambda w^2} \cdot \exp\left({-\frac{\lambda w^{3 - \tau}}{2^d}}\right) \cdot n.
  \end{align*}
\end{lemma}
\begin{proof}
  \begin{figure}[t]
    \centering \includegraphics{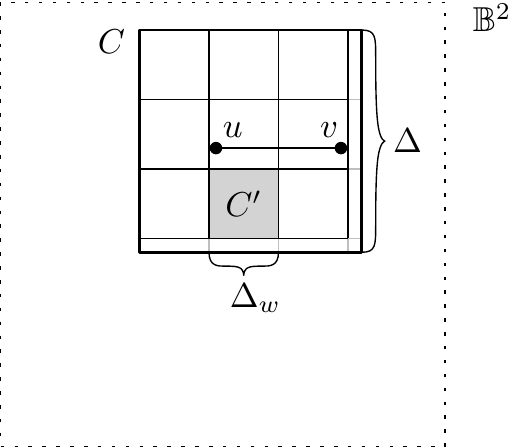}
    \caption{The cell $C$ of width $\Delta$ is divided into sub-cells
      of width $\Delta_w$.  The sub-cell $C'$ is completely contained
      in $C$.  The vertices $u$ and $v$ are in adjacent sub-cells and
      are therefore adjacent themselves.}
    \label{fig:sub-cells}
  \end{figure}
  We discretize the cell $C$ into sub-cells, such that vertices in
  adjacent sub-cells are adjacent themselves, as shown in
  Figure~\ref{fig:sub-cells}.  Note that two vertices $u, v$ with
  weights $w_u, w_v \ge w$ are adjacent if their distance is bounded
  by
  \begin{equation*}
    \dist(u, v) \le
    \left(\frac{\lambda w^2}{n}\right)^{1/d}.
  \end{equation*}
  Thus, all vertices in adjacent sub-cells are guaranteed to be
  adjacent, if the side length of a sub-cell is
  \begin{equation*}
    \Delta_w = \frac{1}{2}\left(\frac{\lambda w^2}{n}\right)^{1/d}.
  \end{equation*}
  Note that for very large $w$, we get $\Delta_w \ge \Delta$, in which case
  all vertices in $C$ are pairwise adjacent with probability 1.  In
  the following, we therefore assume that $w$ is smaller.  For a given
  sub-cell $C'$, we compute the probability for a given vertex $v$ to
  lie in $C'$ as
  \begin{align*}
    \Pr[v \in V(C')] &= (\Delta_w)^d = \left(\frac{1}{2}\left(\frac{\lambda w^2}{n}\right)^{1/d}\right)^{d} = \frac{\lambda w^2}{2^d n}.
  \end{align*}
  Additionally, the probability for $v$ to have weight at least
  $w_v \ge w$, is given by
  \begin{align*}
    \Pr[w_v \ge w] = 1 - \Pr[w_v \le w] = w^{-(\tau - 1)}.
  \end{align*}
  Together, we obtain
  \begin{align*}
    \Pr[v \in V(C') \land w_v \ge w] &= \Pr[v \in V(C')] \cdot \Pr[w_v \ge w] = \frac{\lambda w^2}{2^d n} \cdot w^{-(\tau - 1)} = \frac{\lambda w^{3 - \tau}}{2^d n}.
  \end{align*}
  Consequently, the expected number of vertices of weight at least $w$
  in $C'$ is
  \begin{align*}
    \mathbb{E}\left[\left| \left\{ v \in V(C') \mid w_v \ge w \right\} \right| \right] = \frac{\lambda w^{3 - \tau}}{2^d}.
  \end{align*}
  Since the vertices are distributed according to a Poisson
  distribution, the probability for $C'$ to not contain any of these
  vertices is given by
  \begin{align*}
    \Pr[\left\{ v \in V(C') \mid w_v \ge w \right\} = \emptyset] = \exp\left({-\frac{\lambda w^{3 - \tau}}{2^d}}\right).
  \end{align*}
  Finally, we lower-bound the probability for the vertices of weight
  at least $w$ in our initial cell~$C$ to induce a connected graph, by
  considering the probability that none of its sub-cells is empty.
  Note that we have
  \begin{align*}
    k = \left(\floor*{\frac{\Delta}{\Delta_w}}\right)^d
  \end{align*}
  sub-cells $C_1', \dots, C_k'$ that are \emph{completely} contained
  in the cell $C$.  Clearly, whether the remaining sub-cells
  (intersecting the boundary of $C$) are empty or not has no impact on
  the connectedness of the considered subgraph.  The probability for
  all of the sub-cells $C_1', \dots, C_k'$ to be non-empty can be
  bounded using Bernoulli's inequality, which yields
  \begin{align*}
    \Pr[\forall C' \in \{C_1', \dots, C_k'\} \colon V(C') \neq
    \emptyset]
    &= (1 - \Pr[V(C') = \emptyset])^k \\
    &\ge 1 - k \cdot \Pr[V(C') = \emptyset] \\
    &= 1 - \left(\floor*{\frac{\Delta}{\Delta_w}}\right)^d \cdot \exp\left({-\frac{\lambda w^{3 - \tau}}{2^d}}\right) \\
    &\ge 1 - \Delta^d \cdot \frac{2^d n}{\lambda w^2} \cdot \exp\left({-\frac{\lambda w^{3 - \tau}}{2^d}}\right) \\
    &= 1 - \frac{(2\Delta)^d}{\lambda w^2} \cdot
      \exp\left({-\frac{\lambda w^{3 - \tau}}{2^d}}\right) \cdot n. \qedhere
  \end{align*}
\end{proof}

The following lemma shows that we basically get an independent
coin-flip with non-vanishing success probability for each cell to be
nice.  We want to point out three technical details of the lemma
statement here.  First, the lemma specifically considers the connected
component containing a vertex of weight at least $\hat w$.  We will
later choose $\hat w = \sqrt{n / \lambda}$, i.e., this vertex is part
of the core.  As all core vertices form a clique, this makes sure that
the components we get for the individual cells actually connect to one
large component in the whole graph.  Secondly, the lower bound on
$\mu$ requires that the cells are sufficiently large to contain a
vertex of weight $\hat w$ with non-vanishing probability.  Thirdly,
the lower bound on $\hat w$ ensures that vertices with higher weight
are likely connected by Lemma~\ref{lem:connected-core}.

\begin{lemma}
  \label{lem:coinflip-for-each-cell}
  Let $G \sim \mathcal{G}(n, \mathbb B^d, \tau, \lambda)$ be a
  threshold GIRG, let $\hat w$ be a weight and let $C$ be a cell of
  side length $\Delta$.  Let $\mu = \Delta^d n$ be the expected number
  of vertices in $C$.  If $\mu \ge \hat w^{\tau - 1}$,
  $\mu \in \omega((\log n)^{2 / (3 - \tau) } \log\log(n)^d)$, and
  $\hat w \in \omega((\log n)^{1 / (3 - \tau)})$, then, with
  non-vanishing probability, the graph induced by the vertices in $C$
  contains a vertex of weight at least $\hat w$ whose connected
  component has size $\Theta(\mu)$.
\end{lemma}
\begin{proof}
  The overall argument goes as follows.  First, the lower bound on
  $\mu$ ensures that $C$ contains a vertex of weight $\hat w$ with
  non-vanishing probability.  For a smaller weight
  $\overline w \le \hat w$, we then apply
  Lemma~\ref{lem:connected-core} to get that all vertices of weight at
  least $\overline w$ form a connected component asymptotically almost
  surely.  Afterwards, it remains to show that enough vertices of
  lower weight connect to a vertex of weight at least $\overline w$
  via paths not leaving $C$.  For the existence of these paths, we use
  Lemma~\ref{lem:layer-path}.  To show that most of them do not leave
  $C$, we use that the considered vertices have weight at most
  $\overline w$ and thus cannot deviate too much from the starting
  position.
	
  Recall that the weight of a vertex is at least $\hat w$ with
  probability $\hat w^{-(\tau - 1)}$.  Thus, the expected number of
  vertices in cell $C$ with weight at least $\hat w$ is
  $\mu \hat w^{-(\tau - 1)}$.  Plugging in the bound
  $\mu \ge \hat w^{\tau - 1}$, everything cancels and we obtain an
  expected value of $1$.  As the number of vertices in $C$ with weight
  above $\hat w$ follows a Poisson distribution, we get at least one
  such vertex with non-vanishing probability.
	
  We set $\overline w = (2^d / \lambda \log n)^{1 / (3 - \tau)}$.
  Note that by the condition on $\hat w$ in the lemma statement, we
  have $\overline w \le \hat w$.  Note further that $\overline w$ is
  chosen such that the exponent in the bound of
  Lemma~\ref{lem:connected-core} simplifies to $- \log n$.  Thus by
  Lemma~\ref{lem:connected-core}, the graph induced by the vertices of
  weight at least $\overline w$ in $C$ is connected with probability
  at least $1 - (2 \Delta)^d / (\lambda \overline w^2)$.  As
  $\Delta \le 1$ and $\hat w$ is increasing with $n$, this goes to $1$
  for $n \to \infty$.
	
  Consider a vertex of weight below $\overline w$.  Then, by
  Lemma~\ref{lem:layer-path}, it has a layer path to a vertex with
  weight at least $\overline w$ with non-vanishing probability.  In
  the following, with \emph{layer path} we always refer to a layer
  path that ends in the layer belonging to $\overline w$.  Note that a
  layer path has length at most $O(\log\log n)$.  Also note that the
  largest weight we encounter is in $O(\overline w)$ as the path stops
  in the layer belonging to weight $\overline w$ and the weights
  increase only by a constant factor between layers.  It follows that,
  in each dimension, the distance between two consecutive vertices on
  a layer path is in $O((\overline w^2 / n)^{1/d})$, as the vertices
  would not be connected otherwise.  Thus, the overall deviation of a
  layer path from the starting point is upper bonded by
  $O((\overline w^2 / n)^{1/d} \log\log n) = O\left(((\log n)^{2 / (3
      - \tau)} \log\log(n)^d / n)^{1/d} \right)$.  By the second lower
  bound on $\mu$, this is asymptotically less than $\Delta$.  Thus,
  shrinking $C$ accordingly from all directions yields a subregion
  $C'$ that contains $\Theta(\mu)$ vertices in expectation such that
  any layer path that starts in $C'$ stays in $C$.
	
  Instead of counting all vertices in $C'$ that have layer paths, we
  only count vertices in the first layer.  This has the advantage,
  that the event that an individual vertex in the first layer has a
  layer path is independent of the number of vertices in the first
  layer (while it depends on the number of vertices in higher layers).
  First note that the number of vertices in the first level of $C'$ is
  a random variable following a Poisson distribution with expected
  value in $\Theta(\mu)$.  Thus, there are $\Theta(\mu)$ such vertices
  with non-vanishing probability.
	
  Now let $X \in [0, 1]$ be the random variable that describes the
  fraction of vertices in the first layer that fail to have a layer
  path.  By Lemma~\ref{lem:layer-path}, the probability for an
  individual vertex to not have a layer path is a upper bounded
  constant $p < 1$ (i.e., the layer path exists with non-vanishing
  probability at least $1 - p$).  Thus, we get $\mathbb{E}[X] \le p$.
  Markov's inequality then gives us $\Prob{X \ge c} \le p / c$ and
  thus $\Prob{X < c} \ge 1 - p / c$.  We can choose $c$ to be a
  constant with $p < c < 1$, which gives us a non-vanishing
  probability that a fraction of at least $1 - c > 0$ vertices have
  the desired layer path.  Note that this holds independently of the
  number of vertices actually sampled in the first layer of $C'$.
	
  To wrap up, consider the three events that there exists a vertex of
  weight at least $\hat w$, that there are $\Theta(\mu)$ vertices in
  the first layer of $C'$, and that a constant fraction of them have
  layer paths.  Note that the three events are independent and each
  holds with non-vanishing probability.  Thus, their intersection,
  which we denote with $A$, also holds with non-vanishing probability.
  Finally, the event $B$ that all vertices of weight at least
  $\overline w$ induce a connected graph holds asymptotically almost
  surely.  Though $A$ and $B$ are not independent, we can apply the
  union bound to their complements to obtain that $A$ and $B$ together
  hold with non-vanishing probability.
\end{proof}

\subsection{Large Components are Likely to Exist}
\label{sec:giant-lower}

To obtain the following theorem, it remains to apply a Chernoff bound
to the coin flips obtained for each cell by
Lemma~\ref{lem:coinflip-for-each-cell}.

\begin{theorem}
  \label{thm:giant-likely-to-exist}
  Let $G \sim \mathcal G(n, \mathbb B^d, \tau, \lambda)$ be a
  threshold GIRG.  Then $G$ has a connected component of size
  $\Theta(n)$ with probability
  $1 - \exp\left( - \Omega(n^{(3 - \tau) / 2}) \right)$.
\end{theorem}
\begin{proof}
  First note that the probability to have $\omega(n)$ vertices is
  exponentially small and thus we only have to show the lower bound on
  the size of the largest connected component.  To apply
  Lemma~\ref{lem:coinflip-for-each-cell}, we choose the cell width
  $\Delta$ such that $\Delta^d n = \mu = \hat w^{\tau - 1}$.  With
  this, we obtain that the number of cells $k$ is
  \begin{equation*}
    k \in \Theta\left(\frac{1}{\Delta^d}\right) =
    \Theta\left(\frac{n}{\hat w^{\tau - 1}}\right) =
    \Theta\left(\frac{n}{\left(\sqrt{n/\lambda}\right)^{\tau - 1}}\right) =
    \Theta\left(n^{\frac{3 - \tau}{2}}\right).
  \end{equation*}
  As the chosen $\Delta$ satisfies the conditions of
  Lemma~\ref{lem:coinflip-for-each-cell}, the graph induced by each
  cell contains a vertex from the core whose connected component has
  size $\Theta(\mu)$ with non-vanishing probability.  If this holds
  for a constant fraction of cells, we get a giant component, as all
  vertices of weight at least $\hat w$ form a clique in $G$.  Thus, we
  have $k$ independent coin flips, each succeeding with a probability
  of $p > 0$, and we are interested in the number of successes $X$.
  To show that $X \in \Theta(k)$ is highly likely, we can simply apply
  a Chernoff bound (see~\cite[Theorem 4.4]{mu-pc-05}).  For
  $\delta \ge 0$, we get
  \begin{equation*}
    \Prob{X \le (1 - \delta) \mathbb{E}[X]} \le \exp\left( -
      \frac{\delta^2}{2} \mathbb{E}[X] \right).
  \end{equation*}
  As $\mathbb{E}[X] \in \Theta(k)$, this implies
  $\Prob{X \in o(k)} \le \exp\left( - \Omega(k) \right)$.  Inserting
  $k$ yields the claim.
\end{proof}

As already mentioned in Section~\ref{sec:preliminaries}, this directly
implies the following corollary.

\begin{corollary}
  Theorem~\ref{thm:giant-likely-to-exist} also holds with the torus
  $\mathbb T$ as ground space.
\end{corollary}

The following lemma states a well known property of GIRGs.  As we are
not aware of a formal proof in the literature, we provide one in
Appendix~\ref{sec:proof-that-subgraph-is-girg}.

\begin{restatable}[folklore]{lemma}{subgraInCellIsGIRG}
  Let $H \sim \mathcal{G}(n, \mathbb B^d, \tau, \lambda, T)$ be a GIRG
  and let $G$ be the subgraph induced by the vertices within a cell of
  side length $\Delta = (f(n) / n)^{1 / d}$.  Then
  $G \sim \mathcal{G}(f(n), \mathbb B^d, \tau, \lambda, T)$.
\end{restatable}

Together with Theorem~\ref{thm:giant-likely-to-exist} this yields the
following corollary.  We note that this also yields large connected
components within cells that are too small to contain a core vertex.
For such cells, we know that we get a large connected component but we
do not know whether it connects to the giant of the whole graph.
Clearly, the same statement holds with the torus $\mathbb T^d$ as
ground space.

\begin{corollary}
\label{corollary::partial_giant}
Let $H \sim \mathcal{G}(n, \mathbb{B}^d, \tau, \lambda)$ be a
threshold GIRG and let $G$ be the subgraph induced by the vertices
within a cell of side length $\Delta = (f(n) / n)^{1/d}$.  Then $G$
has a connected component of size $\Theta(f(n))$ with probability
$1 - \exp\left( - \Omega(f(n)^{(3 - \tau) / 2}) \right)$.
\end{corollary}

\section{Conclusion}

Our proof for the emergence of a giant component in geometric
inhomogeneous random graphs builds on three simple arguments.  First,
GIRGs are likely to contain a clique of high-weight vertices.  Second,
the remaining vertices are sufficiently likely to connect to this core
via layer-paths, whose vertices have exponentially increasing weight.
And, third, most of these paths exist sufficiently independently from
each other.

We note that the same argumentation also works for the closely related
hyperbolic random graph model, where the discretization into weight
layers translates to a natural discretization of the underlying
geometric space that was previously used to bound the diameter of
HRGs~\cite{fk-dhrg-18}.

Our resulting strong probability bound can be combined with a simple
coupling argument to identify connected subgraphs of arbitrary size in
certain subregions of the geometric ground space.  In particular, when
these subregions are the cells of a regular grid (as used several
times throughout the paper), we obtain connected subgraphs of roughly
equal size.  We believe that this property can be utilized in the
context of problems with connectivity constraints.  For example, in
the previously mentioned \emph{balanced connected partitioning}
problem~\cite{cc-aambcgp-21, s-ogcp-22}, the goal is to partition the
vertices of a graph into a given number of sets of approximately equal
size, such that their induced subgraphs are connected.  Moreover, in
\emph{component order connectivity}~\cite{ghi-scocm-13} the aim is to
find a minimum number of vertices such that after their removal each
connected component has bounded size.  We conjecture that our
structural insights in Corollary~\ref{corollary::partial_giant} may
prove useful in obtaining efficient algorithms for these problems on
GIRGs and the networks they represent well.

\bibliography{main}

\appendix

\section{Proof That the Subgraph of a Cell is a GIRG}
\label{sec:proof-that-subgraph-is-girg}

\subgraInCellIsGIRG*
\begin{proof}
  Note that we basically consider two ways to generate a graph and
  claim that they give the same probability distribution over graphs.
  Intuitively, this can be seen by generating points with weights in
  the cell $[0, \Delta]^d$, scaling it to the full ground space
  $[0, 1]^d$, and making three observations.  First, for the vertex
  positions, this is equivalent to directly sampling points in
  $[0, 1]^d$.  Secondly, the weight distribution is independent of the
  number of vertices.  Thirdly, the connection probabilities between
  vertices are the same in the scaled variant as they are in the cell.
  To make this more formal, draw $G$ as a subgraph of $H$ as stated in
  the lemma and draw
  $G' \sim \mathcal{G}(f(n), \mathbb B^d, \tau, \lambda, T)$.  We show
  that $G$ and $G'$ follow the same distribution.
	
  Recall that we consider the Poisson variant of the GIRG model, i.e.,
  the vertices are the result of a Poisson point process in the
  product space $\mathbb B^d \times \mathcal W$.  Thus, the vertex set
  of $G$ can be generated by first determining the number of points
  $n_G$ with positions in $[0, \Delta]^d$, which is a random variable
  following a Poisson distribution with expectation
  $n \cdot \Delta^d = f(n)$.  Then, independently for each of the
  $n_G$ vertices, a position is drawn uniformly at random
  from~$[0, \Delta]^d$ and a weight is drawn from $(1, \infty)$ with
  probability density function $(\tau - 1)\cdot w^{-\tau}$.
	
  To generate $G'$, we can also first determine the number of points
  $n_{G'}$, which is also Poisson distributed with expectation $f(n)$.
  Thus, we can couple $n_G$ and $n_{G'}$ to have the same value and we
  assume a one-to-one correspondence between the vertices in $G$ and
  $G'$ in the following.  For each vertex, the weight is again a
  random variable with density $(\tau - 1)\cdot w^{-\tau}$, which only
  depends on $\tau$.  Thus, for each vertex, we can couple its weight
  in $G$ with its weight in $G'$ to assume them to be equal.  The
  position in $G'$ is drawn uniformly from $[0, 1]^d$.  Thus, we can
  couple the random variables for the positions in $G'$ with those in
  $G$ such that a vertex with position $x \in [0, \Delta]^d$ in $G$
  has position $x / \Delta$ in $G'$.  Note that this has the effect
  that all distances between vertices in $G'$ are scaled by a factor
  of $1 / \Delta$ compared to the corresponding distance in $G$.
	
  It remains to show that for every vertex pair $u, v$ the connection
  probability in $G$ is the same as in $G'$.  Let $w_u$ and $w_v$ be
  the weight of $u$ and $v$ (which is the same for $G$ and~$G'$ due to
  the coupling).  Also, let $\dist(u, v)$ be the distance between $u$
  and $v$ in $G$ and let $\dist'(u, v) = \dist(u, v) / \Delta$ be
  their distance in $G'$.  Then (for $T > 0$) the connection
  probability of $u$ and $v$ in $G$ is
  \begin{equation*}
    \Prob{\{u, v\} \in E} = \min\left\{\left( \frac{\lambda w_u w_v}{n \dist(u, v)^d} \right)^{1 / T}, 1\right\}.
  \end{equation*}
  The two things that change for $G'$ is that $n$ is replaced by
  $f(n)$ and $\dist(u, v)^d$ is replaced by
  $\dist'(u, v)^d = (\dist(u, v) / \Delta)^d = n / f(n) \cdot
  \dist(u, v)^d$.  The $f(n)$ cancels out, yielding the same connection
  probability for $G$ and $G'$.  For $T = 0$, the argument works
  analogously. 
\end{proof}

\end{document}